\documentclass{llncs}
\usepackage{multicol}

\usepackage{amsmath}
\usepackage{amsfonts}
\usepackage{amssymb}
\usepackage{verbatim}

\usepackage{graphicx}
\usepackage{epic}
\usepackage{eepic}
\usepackage{epsfig,float}
\usepackage{pdfsync}

\usepackage{multicol}
\renewcommand{\le}{\leqslant}
\renewcommand{\ge}{\geqslant}

\newcommand{\ol}{\overline}
\newcommand{\eps}{\varepsilon}
\newcommand{\emp}{\emptyset}

\newcommand{\Sig}{\Sigma}

\newcommand{\noin}{\noindent}

\newcommand{\bi}{\begin{itemize}}
\newcommand{\ei}{\end{itemize}}
\newcommand{\be}{\begin{enumerate}}
\newcommand{\ee}{\end{enumerate}}
\newcommand{\bd}{\begin{description}}
\newcommand{\ed}{\end{description}}
\newcommand{\bq}{\begin{quote}}
\newcommand{\eq}{\end{quote}}

\newcommand{\cA}{{\mathcal A}}

\newcommand{\cD}{{\mathcal D}}

\newcommand{\cN}{{\mathcal N}}

\newcommand{\cT}{{\mathcal T}}

\newcommand{\rev}{\mathbb{R}}
\newcommand{\deter}{\mathbb{D}}

\title{Quotient Complexities of Atoms of Regular Languages\thanks{This work was supported by the Natural Sciences and Engineering Research Council of Canada under grant No.~OGP0000871, by the Estonian Center of Excellence in Computer Science, EXCS, financed by the European Regional Development Fund, and by the Estonian Science Foundation grant 7520.}
}
\author{Janusz~Brzozowski\inst{1} \and Hellis~Tamm\inst{2}}

\authorrunning{Brzozowski, Tamm}   

\institute{David R. Cheriton School of Computer Science, University of Waterloo, \\
Waterloo, ON, Canada N2L 3G1\\
\{{\tt brzozo@uwaterloo.ca}\}
\and
Institute of Cybernetics, Tallinn University of Technology,\\
Akadeemia tee 21, 12618 Tallinn, Estonia\\
\{{\tt hellis@cs.ioc.ee}\} 
}
\begin{document}
\maketitle
%\today
\begin{abstract}
An atom  of a regular language $L$ with $n$ (left) quotients is a non-empty 
intersection of uncomplemented or complemented quotients of $L$, where each 
of the $n$ quotients appears in a term of the intersection.
The quotient complexity of $L$, which is the same as the state complexity of $L$, 
is the number of quotients of $L$.
We prove that, for any language $L$ with quotient complexity $n$, 
the quotient complexity of any atom of $L$ with $r$ complemented 
quotients has an upper bound of $2^n-1$ if $r=0$ or $r=n$, and 
$1+\sum_{k=1}^{r} \sum_{h=k+1}^{k+n-r} C_{h}^{n} \cdot C_{k}^{h}$ 
otherwise, where $C_j^i$ is the binomial coefficient. 
For each $n\ge 1$, we exhibit a language whose atoms meet these bounds.

\end{abstract}

\section{Introduction}
\label{sec:intro}

Atoms of regular languages were introduced in 2011 by Brzozowski and Tamm~\cite{BrTa11};
we briefly state their main properties here.

 The \emph{(left) quotient} of a regular language $L$ over an 
alphabet $\Sig$ by a word $w\in\Sig^*$ is the language 
$w^{-1}L=\{x\in\Sig^*\mid wx\in L\}$.
It is well known that a language $L$ is regular if and only if it has a finite number of distinct quotients, and that the number of states in the minimal deterministic finite automaton (DFA) recognizing $L$ is precisely the number of distinct quotients of $L$. 
Also, $L$ is its own quotient by the empty word $\eps$, that is $\eps^{-1}L=L$.
Note too that the quotient by $u\in\Sig^*$ of the quotient by $w\in\Sig^*$ of $L$ is the quotient  by $wu$
of $L$, that is, $u^{-1}(w^{-1}L)=(wu)^{-1}L $.

An \emph{atom}\footnote{The definition in \cite{BrTa11} does not consider 
the intersection of all the complemented quotients to be an atom. 
Our new definition adds symmetry to the theory.}
of a regular language $L$ with quotients $K_0,\ldots, K_{n-1}$ is any 
non-empty language of the form 
$\widetilde{K_0}\cap \cdots \cap \widetilde{K_{n-1}}$, 
where $\widetilde{K_i}$ is either $K_i$ or $\ol{K_i}$, and $\ol{K_i}$ is the complement of $K_i$ with respect to $\Sig^*$. 
Thus atoms of $L$ are regular languages uniquely determined by $L$ and
they define a partition of $\Sig^*$. 
They are pairwise disjoint,  every quotient of $L$ (including $L$ itself)  
is a union of atoms, and  every quotient of an atom is a union of atoms.
Thus the atoms of a regular language are its basic building blocks. 
Also, $\ol{L}$ defines the same atoms as   $L$. 

The \emph{quotient complexity}~\cite{Brz10} of $L$ is the number 
of quotients of $L$, and this is the same number as the number of 
states in the minimal DFA recognizing $L$; 
the latter number is known as the \emph{state complexity}~\cite{Yu01} 
of $L$. 
Quotient complexity allows us to use  language-theoretic 
methods, whereas state complexity is more amenable to automaton-theoretic 
techniques. We use one of these two points of view or the other, 
depending on convenience.

We study the quotient complexity of atoms of regular languages.
Suppose that $L\subseteq\Sig^*$ is a non-empty regular language 
and its set of quotients is 
 $K=\{K_0, K_1, \ldots, K_{n-1}\}$, with $n\ge 1$. 
Our main result is the following:
\begin{theorem}[Main Result]
\label{thm:main}

For $n\ge 1$, the quotient complexity of the atoms with 0 or $n$ complemented quotients 
is less than or equal to   $2^n-1$.
For $n\ge 2$ and $r$ satisfying $1\le r\le n-1$, the quotient complexity of any atom of $L$ with $r$ complemented quotients is less than or equal~to 
\begin{equation*}
f(n,r)=1 + \sum_{k=1}^{r} \sum_{h=k+1}^{k+n-r} C_{h}^{n} \cdot C_{k}^{h}.
\end{equation*}
For $n=1$, the single atom $\Sig^*$ of the language $\Sig^*$ or $\emp$ meets the bound 1.
Moreover, for $n\ge 2$, all the atoms of the language $L_n$ recognized by the DFA  $\cD_n$ of 
Figure~\ref{fig:witness} meet these bounds.
\end{theorem}

\begin{figure}[hbt]
\begin{center}
\setlength{\unitlength}{0.00043745in}
\begingroup\makeatletter\ifx\SetFigFont\undefined%
\gdef\SetFigFont#1#2#3#4#5{%
  \reset@font\fontsize{#1}{#2pt}%
  \fontfamily{#3}\fontseries{#4}\fontshape{#5}%
  \selectfont}%
\fi\endgroup%
{\renewcommand{\dashlinestretch}{30}
\begin{picture}(7318,1617)(0,-10)
\put(15,659){\makebox(0,0)[lb]{\smash{{\SetFigFont{9}{10.8}{\familydefault}{\mddefault}{\updefault}$\cD_n$}}}}
\put(3388.000,1151.333){\arc{333.333}{2.2143}{7.2105}}
\blacken\path(3525.638,1114.417)(3488.000,1018.000)(3566.107,1085.913)(3525.638,1114.417)
\put(2308.000,1151.333){\arc{333.333}{2.2143}{7.2105}}
\blacken\path(2445.638,1114.417)(2408.000,1018.000)(2486.107,1085.913)(2445.638,1114.417)
\put(5820.000,1158.333){\arc{333.333}{2.2143}{7.2105}}
\blacken\path(5957.638,1121.417)(5920.000,1025.000)(5998.107,1092.913)(5957.638,1121.417)
\put(6953.000,1211.333){\arc{333.333}{2.2143}{7.2105}}
\blacken\path(7090.638,1174.417)(7053.000,1078.000)(7131.107,1145.913)(7090.638,1174.417)
\put(1232,703){\ellipse{630}{630}}
\put(3396,703){\ellipse{630}{630}}
\put(2326,701){\ellipse{630}{630}}
\put(5828,703){\ellipse{630}{630}}
\put(6951,702){\ellipse{630}{630}}
\put(6950,696){\ellipse{720}{720}}
\path(653,703)(923,703)
\blacken\path(803.000,673.000)(923.000,703.000)(803.000,733.000)(803.000,673.000)
\path(2633,703)(3083,703)
\blacken\path(2963.000,673.000)(3083.000,703.000)(2963.000,733.000)(2963.000,673.000)
\path(3713,703)(4163,703)
\blacken\path(4043.000,673.000)(4163.000,703.000)(4043.000,733.000)(4043.000,673.000)
\path(1553,703)(2003,703)
\blacken\path(1883.000,673.000)(2003.000,703.000)(1883.000,733.000)(1883.000,673.000)
\path(5063,711)(5513,711)
\blacken\path(5393.000,681.000)(5513.000,711.000)(5393.000,741.000)(5393.000,681.000)
\path(6136,703)(6586,703)
\blacken\path(6466.000,673.000)(6586.000,703.000)(6466.000,733.000)(6466.000,673.000)
\path(2138,965)(2137,967)(2134,970)
	(2129,977)(2121,986)(2110,998)
	(2097,1014)(2082,1031)(2065,1049)
	(2046,1068)(2026,1088)(2005,1106)
	(1982,1125)(1959,1141)(1933,1157)
	(1906,1170)(1877,1182)(1846,1191)
	(1813,1196)(1778,1198)(1744,1195)
	(1711,1188)(1680,1178)(1652,1165)
	(1627,1150)(1603,1134)(1581,1116)
	(1560,1097)(1541,1077)(1522,1057)
	(1505,1037)(1490,1018)(1476,1000)
	(1465,985)(1455,972)(1440,950)
\blacken\path(1482.814,1066.047)(1440.000,950.000)(1532.387,1032.247)(1482.814,1066.047)
\path(6638,493)(6637,493)(6636,492)
	(6634,491)(6630,489)(6624,487)
	(6617,483)(6607,479)(6595,474)
	(6581,468)(6564,460)(6544,452)
	(6522,442)(6497,432)(6470,420)
	(6440,408)(6407,395)(6372,381)
	(6336,367)(6297,352)(6256,337)
	(6213,321)(6168,305)(6122,289)
	(6074,273)(6024,257)(5973,241)
	(5921,226)(5866,210)(5810,195)
	(5752,180)(5692,165)(5630,151)
	(5566,137)(5500,123)(5430,110)
	(5358,98)(5283,86)(5205,75)
	(5123,64)(5038,54)(4949,45)
	(4856,37)(4759,30)(4660,24)
	(4556,19)(4451,15)(4343,13)
	(4246,12)(4150,13)(4053,15)
	(3958,17)(3864,21)(3772,25)
	(3682,30)(3593,37)(3507,43)
	(3423,51)(3340,59)(3260,67)
	(3181,76)(3104,86)(3028,96)
	(2954,106)(2881,117)(2810,128)
	(2740,139)(2671,151)(2603,163)
	(2536,175)(2470,188)(2405,200)
	(2342,213)(2279,226)(2218,239)
	(2158,252)(2100,265)(2043,277)
	(1988,290)(1935,302)(1883,314)
	(1835,326)(1788,337)(1744,348)
	(1703,358)(1664,367)(1629,376)
	(1597,384)(1567,391)(1541,398)
	(1518,404)(1499,409)(1482,413)
	(1468,417)(1457,420)(1448,422)
	(1442,424)(1433,426)
\blacken\path(1556.650,429.254)(1433.000,426.000)(1543.635,370.683)(1556.650,429.254)
\put(1177,636){\makebox(0,0)[lb]{\smash{{\SetFigFont{7}{8.4}{\rmdefault}{\mddefault}{\updefault}$0$}}}}
\put(2265,636){\makebox(0,0)[lb]{\smash{{\SetFigFont{7}{8.4}{\rmdefault}{\mddefault}{\updefault}$1$}}}}
\put(3352,636){\makebox(0,0)[lb]{\smash{{\SetFigFont{7}{8.4}{\rmdefault}{\mddefault}{\updefault}$2$}}}}
\put(2761,801){\makebox(0,0)[lb]{\smash{{\SetFigFont{7}{8.4}{\familydefault}{\mddefault}{\updefault}$a$}}}}
\put(3833,823){\makebox(0,0)[lb]{\smash{{\SetFigFont{7}{8.4}{\familydefault}{\mddefault}{\updefault}$a$}}}}
\put(6263,823){\makebox(0,0)[lb]{\smash{{\SetFigFont{7}{8.4}{\familydefault}{\mddefault}{\updefault}$a$}}}}
\put(1598,838){\makebox(0,0)[lb]{\smash{{\SetFigFont{7}{8.4}{\familydefault}{\mddefault}{\updefault}$a,b$}}}}
\put(5558,658){\makebox(0,0)[lb]{\smash{{\SetFigFont{6}{7.2}{\familydefault}{\mddefault}{\updefault}$n-2$}}}}
\put(4516,644){\makebox(0,0)[lb]{\smash{{\SetFigFont{7}{8.4}{\familydefault}{\mddefault}{\updefault}$\cdots$}}}}
\put(5191,816){\makebox(0,0)[lb]{\smash{{\SetFigFont{7}{8.4}{\familydefault}{\mddefault}{\updefault}$a$}}}}
\put(6825,1431){\makebox(0,0)[lb]{\smash{{\SetFigFont{7}{8.4}{\familydefault}{\mddefault}{\updefault}$b$}}}}
\put(4005,133){\makebox(0,0)[lb]{\smash{{\SetFigFont{7}{8.4}{\familydefault}{\mddefault}{\updefault}$a,c$}}}}
\put(6681,640){\makebox(0,0)[lb]{\smash{{\SetFigFont{6}{7.2}{\familydefault}{\mddefault}{\updefault}$n-1$}}}}
\put(5603,1408){\makebox(0,0)[lb]{\smash{{\SetFigFont{7}{8.4}{\familydefault}{\mddefault}{\updefault}$b,c$}}}}
\put(1680,1273){\makebox(0,0)[lb]{\smash{{\SetFigFont{7}{8.4}{\familydefault}{\mddefault}{\updefault}$b$}}}}
\put(3188,1423){\makebox(0,0)[lb]{\smash{{\SetFigFont{7}{8.4}{\familydefault}{\mddefault}{\updefault}$b,c$}}}}
\put(2198,1423){\makebox(0,0)[lb]{\smash{{\SetFigFont{7}{8.4}{\familydefault}{\mddefault}{\updefault}$c$}}}}
\put(1110,1401){\makebox(0,0)[lb]{\smash{{\SetFigFont{7}{8.4}{\familydefault}{\mddefault}{\updefault}$c$}}}}
\put(1233.000,1136.333){\arc{333.333}{2.2143}{7.2105}}
\blacken\path(1370.638,1099.417)(1333.000,1003.000)(1411.107,1070.913)(1370.638,1099.417)
\end{picture}
}
\end{center}
\caption{DFA $\cD_n$ of language $L_n$ whose atoms meet the bounds.} 
\label{fig:witness}
\end{figure}
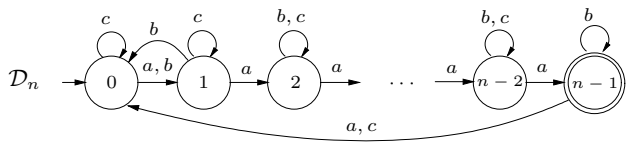

In Section~\ref{sec:bounds} we derive upper bounds on the quotient complexities of atoms.
In Section~\ref{sec:aut} we define our notation and terminology for automata,
and
 present the definition of the \'atomaton~\cite{BrTa11} of a regular language; this is a nondeterministic finite automaton (NFA) whose states are the atoms 
of the language. 
We also provide a different characterization of the \'atomaton.
We introduce a class of DFA's in Section~\ref{sec:witness} and study the \'atomata of 
their languages. We then prove  in Section~\ref{sec:complexity} that the atoms 
of these languages meet the quotient complexity bounds. 
Section~\ref{sec:conclusions} concludes the paper.

\section{Upper Bounds on the Quotient Complexities of Atoms}
\label{sec:bounds}

We first derive upper bounds on the quotient complexity of atoms. 
We use  quotients here, since they are convenient 
for this task. First we deal with the two atoms that have only 
uncomplemented or only complemented quotients.

\begin{proposition}
[Atoms with 0 or $n$ Complemented Quotients]
\label{prop:no_comp}\\
For $n\ge 1$, the quotient complexity of the two atoms $A_K=K_0\cap \cdots\cap K_{n-1}$ 
and $A_\emp = \ol{K_0}\cap \cdots\cap \ol{K_{n-1}}$ is less than or equal to $2^n-1$.
\end{proposition}
\begin{proof}
Every quotient $w^{-1}A_K$ of atom $A_K$ is the intersection of languages 
$w^{-1}K_i$, which  are quotients of $L$:
\begin{eqnarray*}            
w^{-1}A_K &=& w^{-1}(K_0\cap\cdots\cap  {K_{n-1}})
=w^{-1}K_0\cap\cdots\cap {w^{-1}K_{n-1}}. 
\end{eqnarray*}
Since these quotients of $L$ need not be distinct, 
$w^{-1}A_K$ may be the intersection of any non-empty subset of 
quotients of $L$. Hence $A_K$ can have at most $2^n-1$ quotients.

The argument for the atom $A_\emp = \ol{K_0}\cap \cdots\cap \ol{K_{n-1}}$ 
with $n$ complemented quotients is similar, 
since $w^{-1}\ol{K_i}=\ol{w^{-1}K_i}$.
\qed
\end{proof}

Next, we present an upper bound on the quotient complexity of any atom 
with at least one and fewer than $n$ complemented quotients.

\begin{proposition}
[Atoms with $r$ Complemented Quotients, $1\le r\le n-1$]
\label{prop:upperbound}
For $n\ge 2$ and $1\le r\le n-1$, the quotient complexity of any atom 
with $r$ complemented quotients is less than or equal to 
\begin{equation}
f(n,r)=1 + \sum_{k=1}^{r} \sum_{h=k+1}^{k+n-r} C_{h}^{n} \cdot C_{k}^{h},
\end{equation}
where $C_j^i$ is the binomial coefficient 
``$i$ choose $j$".
\end{proposition}
\begin{proof}
Consider an intersection of complemented and uncomplemented quotients 
that constitutes an atom.
Without loss of generality, we arrange the terms in the intersection 
in such a way that all complemented quotients appear on the right. 
Thus let 
$A_i = K_0\cap\cdots\cap K_{n-r-1}\cap\ol{K_{n-r}}\cap\cdots\cap \ol{K_{n-1}}$
be an atom of $L$ with $r$ complemented quotients of $L$, where 
$1\le r\le n-1$.
The quotient of $A_i$ by any word $w\in\Sigma^*$ is 
\begin{eqnarray*}            
w^{-1}A_i  &=& w^{-1}(K_0\cap\cdots\cap K_{n-r-1}\cap\ol{K_{n-r}}\cap\cdots\cap 
\ol{K_{n-1}})\\
&=& w^{-1}K_0\cap\cdots\cap w^{-1}K_{n-r-1}\cap \ol{w^{-1}K_{n-r}}\cap\cdots\cap 
\ol{w^{-1}K_{n-1}}. 
\end{eqnarray*}
Since each quotient $w^{-1}K_j$ is a quotient, say $K_{i_j}$, of $L$, we have
\begin{eqnarray*}
w^{-1}A_i  &=& K_{i_0}\cap\cdots\cap K_{i_{n-r-1}}\cap\ol{K_{i_{n-r}}}\cap\cdots\cap 
\ol{K_{i_{n-1}}}.
\end{eqnarray*}

The cardinality of a set $S$ is denoted by $|S|$.
Let the set of distinct quotients of $L$ appearing in $w^{-1}A_i$ 
uncomplemented (respectively, complemented) be $X$ (respectively, $Y$), where
$1\le |X|\le n-r$ and $1 \le |Y|\le r$.
If $X\cap Y\neq \emp$, then $w^{-1}A_i=\emp$. 
Therefore assume that $X\cap Y = \emp$, and that $|X\cup Y|=h$, 
where $2\le h\le n$; there are $C_h^n$ such sets $X\cup Y$.
Suppose further that $|Y|=k$, where $1\le k\le r$.
There are $C_k^h$ ways of choosing $Y$. Hence there are at most
$\sum_{h=k+1}^{k+n-r} C_{h}^{n} \cdot C_{k}^{h}$
distinct intersections with $k$ complemented quotients.
Thus, the total number of intersections of uncomplemented 
and complemented quotients can be at most
$\sum_{k=1}^{r} \sum_{h=k+1}^{k+n-r} C_{h}^{n} \cdot C_{k}^{h}$.

Adding 1 for the empty quotient of $w^{-1}A_i$, we get the required bound.
\qed
\end{proof}

We now consider the properties of the function $f(n,r)$. 

\begin{proposition}[Properties of Bounds]
For any $n\ge 2$ and $1\le r\le n-1$, 
\be
\item
$f(n,r)=f(n,n-r)$.
\item
For a fixed $n$, the maximal value of $f(n,r)$ occurs when $r=\lfloor n/2 \rfloor$.
\ee
\end{proposition}
\begin{proof}
Since $f(n,r)=1+ \sum_{k=1}^{r} \sum_{h=k+1}^{k+n-r} C_{h}^{n} \cdot C_{k}^{h}$, 
and the following equations hold: 
\begin{eqnarray*}    
\sum_{k=1}^{r} \sum_{h=k+1}^{k+n-r} C_{h}^{n} \cdot C_{k}^{h}=
\sum_{k=1}^{r} \sum_{l=1}^{n-r} C_{k+l}^{n} \cdot C_{k}^{k+l}=
\sum_{l=1}^{n-r} \sum_{k=1}^{r} C_{k+l}^{n} \cdot C_{k}^{k+l}\\
=\sum_{l=1}^{n-r} \sum_{k=1}^{r} C_{k+l}^{n} \cdot C_{l}^{k+l}=
\sum_{l=1}^{n-r} \sum_{m=l+1}^{l+r} C_{m}^{n} \cdot C_{l}^{m},        
\end{eqnarray*}
we have $f(n,r)=f(n,n-r)$.

For the second part,
we will assume that $1\le r <\lfloor n/2 \rfloor$,
and show that $f(n,r+1)>f(n,r)$ for this case.  
We find $f(n,r+1)-f(n,r)$ as follows:
\begin{eqnarray*}    
f(n,r+1)-f(n,r)&=&1+\sum_{k=1}^{r+1} \sum_{h=k+1}^{k+n-r-1} C_{h}^{n} \cdot C_{k}^{h}-
(1+\sum_{k=1}^{r} \sum_{h=k+1}^{k+n-r} C_{h}^{n} \cdot C_{k}^{h})\\
&=&\sum_{k=r+1}^{r+1} \sum_{h=k+1}^{k+n-r-1} C_{h}^{n} \cdot C_{k}^{h}-
\sum_{k=1}^{r} \sum_{h=k+n-r}^{k+n-r} C_{h}^{n} \cdot C_{k}^{h}\\
&=&\sum_{h=r+2}^{n} C_{h}^{n} \cdot C_{r+1}^{h}-
\sum_{k=1}^{r} C_{k+n-r}^{n} \cdot C_{k}^{k+n-r}.
\end{eqnarray*}
Since the first summation can be written as
\begin{eqnarray*}    
\sum_{h=r+2}^{n} C_{h}^{n} \cdot C_{r+1}^{h}&=&
\sum_{h=r+2}^{n-r} C_{h}^{n} \cdot C_{r+1}^{h}+
\sum_{h=n-r+1}^{n} C_{h}^{n} \cdot C_{r+1}^{h}\\
&=&\sum_{h=r+2}^{n-r} C_{h}^{n} \cdot C_{r+1}^{h}+
\sum_{k=1}^{r} C_{k+n-r}^{n} \cdot C_{r+1}^{k+n-r},
\end{eqnarray*}
we get 
\begin{eqnarray*}    
f(n,r+1)-f(n,r)=\sum_{h=r+2}^{n-r} C_{h}^{n} \cdot C_{r+1}^{h}
&+&\sum_{k=1}^{r} C_{k+n-r}^{n} \cdot C_{r+1}^{k+n-r}\\
&-&\sum_{k=1}^{r} C_{k+n-r}^{n} \cdot C_{k}^{k+n-r}.
\end{eqnarray*}

Assuming $1\le k\le r$, we will show that $C_{r+1}^{k+n-r}>C_{k}^{k+n-r}$.
We can express the ratio $C_{r+1}^{k+n-r}/C_{k}^{k+n-r}$ as follows: 
\begin{eqnarray*} 
\frac{C_{r+1}^{k+n-r}}{C_{k}^{k+n-r}}
&=&\frac{(k+n-r)!}{(r+1)!(k+n-2r-1)!}\div\frac{(k+n-r)!}{k!(n-r)!}\\
&=&\frac{k!(n-r)!}{(r+1)!(k+n-2r-1)!}\\
&=&\frac{k!(n-r)\cdots(n-2r+k)(n-2r+k-1)!}{(r+1)\cdots(k+1)k!(n-2r+k-1)!}\\
&=&\frac{(n-r)\cdots(n-2r+k)}{(r+1)\cdots(k+1)}.\\
\end{eqnarray*}
Note that there are $r-k+1$ factors both in the numerator and 
the denominator of the obtained fraction. Therefore, we can write
\begin{eqnarray*} 
\frac{C_{r+1}^{k+n-r}}{C_{k}^{k+n-r}}
=\frac{n-r}{r+1}\cdot\frac{n-r-1}{r}\cdot\cdots\cdot\frac{n-2r+k}{k+1}.
\end{eqnarray*}
The condition $1\le r < \lfloor n/2 \rfloor$ implies that $n>2r+1$; consequently we have
$$n-r > r+1, \: n-r-1> r, \: \ldots, \: n-2r+k > k+1.$$
Therefore $C_{r+1}^{k+n-r}/C_k^{k+n-r}>1$, which implies that 
$C_{r+1}^{k+n-r}>C_k^{k+n-r}$.
 
It follows that 
$$\sum_{k=1}^{r} C_{k+n-r}^{n} \cdot C_{r+1}^{k+n-r}>
\sum_{k=1}^{r} C_{k+n-r}^{n} \cdot C_{k}^{k+n-r},$$
and $f(n,r+1)-f(n,r)>0$.
So, if $1\le r < \lfloor n/2 \rfloor $, then $f(n,r+1)>f(n,r)$.
Since $f(n,r)=f(n,n-r)$, 
the maximum of $f(n,r)$ occurs when $r=\lfloor n/2 \rfloor$.
\qed 
\end{proof}

To better illustrate the properties of $f(n,r)$, we derive explicit formulas 
for the first three values of $r$.
Using the well-known identity
\begin{equation}
\sum_{h=k}^{n} C_{h}^{n} \cdot C_{k}^{h}=2^{n-k}C^n_k,
\end{equation}
we find
\begin{displaymath}
 \begin{array}{rl}
   f(n,1)&=  n2^{n-1}-n+1, \\
   f(n,2)&=n2^{n-1}-2n+\frac{n(n-1)}{ 2}(2^{n-2}-1)+1, \\
   f(n,3)&= n2^{n-1}-(n^2+n) + \frac{n(n-1)(n+4)}{6} (2^{n-3}-1)+1.
 \end{array}
\end{displaymath}

Some numerical values of $f(n,r)$ are shown in Table~\ref{tab:atomcomp}.
The figures in boldface type are the maxima for a fixed $n$.
The row marked \emph{max} shows the maximal quotient complexity of the atoms of $L$.
The row marked \emph{ratio} shows the value of 
$f(n,\lfloor n/2 \rfloor)/f(n-1,\lfloor (n-1)/2 \rfloor)$, for $n\ge 2$. 
It appears that this ratio converges to 3. 
For example, for $n = 100$ it is approximately 3.0002.

\begin{table}[ht]
\caption{Maximal quotient complexity of atoms.}
\label{tab:atomcomp}
\begin{center}
$
\begin{array}{| c| c|c| c|c| c|c| c|c| c|c|c|}    
\hline
\ \ n \ \ 
&\ \ 1 \ \ &\ \ 2 \ \ &  \ 3 \ &\ 4 \ &\ 5 \ &  \ 6 \ &\ 7 \ &\ 8 \ & \ 9 \ &10 &\cdots  
\\
\hline
\hline  
$r=0$
&  \bf 1 &  \bf 3 & \ 7  \ & \ 15 \  & \ 31 \ & \ 63  \  & \  127 \ & \ 255  \  & 511& 1,023  & \cdots  \\
\hline  
$r=1$
&  \bf 1 & \bf 3  & \ \bf 10  \ & \ 29  \  & \ 76 \ & \  187  \  & \  442  \ & \ 1,017 \  & 2,296 &5,111 &\cdots \\
\hline  
$r=2$ 
&  \ast &  \bf 3 & \ \bf 10 \ & \ \bf 43  \  & \ \bf 141  \ & \ 406   \  & \ 1,086  \ & \  2,773 \  & 6,859 &
16,576  &\cdots\\
\hline  
$r=3$
&  \ast & \ast  & \ 7 \ & \  29 \  & \ \bf 141 \ & \ \bf 501  \  & \ \bf  1,548  \ & \  4,425 \  &12,043 &31,681  &\cdots \\
\hline  
$r=4$
&  \ast & \ast  & \ \ast  \ & \  15 \  & \ 76 \ & \ 406  \  & \ \bf 1,548   \ & \ \bf 5,083  \  & \bf 15,361 & 44,071 & \cdots\\
\hline  
$r=5$
&  \ast & \ast  & \ \ast  \ & \  \ast \  & \ 31 \ & \ 187  \  & \  1,086   \ & \ 4,425  \  & \bf 15,361 & \bf 48,733 & \cdots\\
\hline\hline
\ max \
& 1 & 3  & \ 10 \ & \ 43  \  & \ 141 \ & \ 501\  & \ 1,548 \ & \ 5,083  \ & \ 15,361\ &\  48,733 \ &\cdots \\
\hline  
ratio
& - & 3  & \ 3.33 \ & \ 4.30  \  & \ 3.28 \ & \ 3.55 \  & \ 3.09 \ & \ 3.28 \  & 3.02 & 3.17  & \cdots \\
\hline
\end{array}
$
\end{center}
\end{table}

%\vspace{-1cm}

%Before proving that the bounds of Theorem~\ref{thm:main} can be met, 
%we need to derive a number of automaton-theoretic results. 

\section{Automata and \'Atomata of Regular Languages}
\label{sec:aut}

If $\Sig$ is a non-empty finite alphabet, then $\Sig^*$ is 
the free monoid generated by $\Sig$.
A \emph{word} is any element of $\Sig^*$, and the empty word 
is $\eps$. 
A \emph{language} over $\Sig$ is any subset of $\Sig^*$. 
The \emph{reverse of a language} $L$ is denoted 
by $L^R$ and defined as $L^R=\{w^R\mid w\in L\}$, where $w^R$ is $w$ spelled backwards.

A~\emph{nondeterministic finite automaton (NFA)} is a quintuple 
$\cN=(Q, \Sig, \eta, I,F)$, where 
$Q$ is a finite, non-empty set of \emph{states}, 
$\Sig$ is a finite non-empty \emph{alphabet}, 
$\eta:Q\times \Sig\to 2^Q$ is the  \emph{transition function},
$I\subseteq  Q$ is the set of  \emph{initial states},
and $F\subseteq Q$ is the set of \emph{final states}.
As usual, we extend the transition function to functions 
$\eta':Q\times \Sig^*\to 2^Q$, and 
$\eta'':2^Q\times \Sig^*\to 2^Q$.
We do not distinguish these functions notationally, but use 
$\eta$ for all three.

The \emph{language accepted} by an NFA $\cN$ is 
$L(\cN)=\{w\in\Sig^*\mid \eta(I,w)\cap F\neq \emp\}$.
Two NFA's are \emph{equivalent} if they accept the same language. 
The \emph{right language} of a state $q$ of $\cN$ is
$L_{q,F}(\cN)=\{w\in\Sig^* \mid \eta(q,w)\cap F\neq\emp\}$.
The \emph{right language} of a set $S$ of states of $\cN$ is
$L_{S,F}(\cN)=\bigcup_{q\in S} L_{q,F}(\cN)$; hence
$L(\cN)=L_{I,F}(\cN)$.
A~state is \emph{empty} if its right language is empty.
Two states  of an NFA are \emph{equivalent} if their right 
languages are equal. 
The \emph{left language} of a state $q$ of $\cN$ is
$L_{I,q}=\{w\in\Sig^* \mid q\in \eta(I,w)\}$.
A state is \emph{unreachable} if its left language is empty.
An NFA is \emph{trim} if it has no empty or unreachable states.

A \emph{deterministic finite automaton (DFA)} is a quintuple 
$\cD=(Q, \Sig, \delta, q_0,F)$, where
$Q$, $\Sig$, and $F$ are as in an NFA, 
$\delta:Q\times \Sig\to Q$ is the transition function, 
and $q_0$ is the initial state. 
A DFA is an NFA in which the set of initial states is $\{q_0\}$
and the range of $\delta$ is restricted to 
singletons $\{q\}$, $q\in Q$.
Note that an empty state of $\cN$ is an unreachable state of $\cN^\rev$ and vice versa.

We use the following operations on automata: 
\be
\item
The \emph{determinization} operation $\deter$ 
applied to an NFA $\cN$ yields a DFA $\cN^{\deter}$ obtained by 
the well-known subset construction, where only subsets  reachable 
from the initial subset of $\cN^\deter$ are used and the empty subset, 
if present, is included. 
\item
The \emph{reversal} operation $\rev$ applied to an NFA $\cN$ yields 
an NFA $\cN^{\rev}$, where sets of initial and final states of $\cN$ are 
interchanged and each transition between any two states 
is reversed. 
\ee

From now on we consider only non-empty regular languages.
Let $L$ be any such language, and let
its set of quotients be $K=\{K_0,\ldots, K_{n-1}\}$. 
One of the quotients of $L$ is $L$ itself;
this is called the \emph{initial} quotient and is denoted by $K_{in}$.
A quotient is  \emph{final} if it contains the empty word $\eps$.
The set of final quotients is  $F=\{K_i \mid \eps \in K_i\}$.

In the following definition we use a one-to-one correspondence 
$K_i \leftrightarrow  {\mathbf K}_i$ between quotients $K_i$ of 
a language $L$ and the states ${\mathbf K}_i$ of the \emph{quotient DFA} 
$\cD$ defined below.
We refer to the ${\mathbf K}_i$ as \emph{quotient symbols}.

\begin{definition}
\label{def:quotientDFA}
The \emph{quotient DFA} of $L$ is 
$\cD=({\mathbf K},\Sig,\delta, {\mathbf K_{in}}, {\mathbf F})$, where
${\mathbf K}=\{{\mathbf K_0},\ldots,{\mathbf K_{n-1}}\}$,
${\mathbf K_{in}}$ corresponds to $K_{in}$,
${\mathbf F}=\{{\mathbf K_i}\mid K_i \in F\}$, and 
$\delta({\mathbf K}_i, a)={\mathbf K}_j$ if and only if 
$a^{-1}K_i = K_j$, for all ${\mathbf K_i},{\mathbf K_j}\in {\mathbf K}$ and $a\in \Sig$.
\end{definition}

In a quotient DFA
the right language of ${\mathbf K_i}$ is $K_i$, and 
its left language 
is $\{w\in\Sig^*\mid w^{-1}L=K_i\}$.
The latter is the equivalence class of the Nerode equivalence~\cite{Ner58}.
The  language $L(\cD)$  is the right language of ${\mathbf K_{in}}$, and hence 
$L(\cD)=L$.
Also, DFA $\cD$ is minimal, since all quotients in $K$ are distinct.

It follows from the definition of an atom, that a regular language  $L$ has at most $2^n$ atoms. 
An atom is \emph{initial} if it has $L$ (rather than $\ol{L}$) as a term;
it is \emph{final} if it contains~$\eps$.
Since $L$ is non-empty, it has at least one quotient containing~$\eps$. 
Hence it has exactly one final atom, the atom 
$\widehat{K_0}\cap \cdots \cap \widehat{K_{n-1}}$, where 
$\widehat{K_i}=K_i$ if $\eps\in K_i$, and $\widehat{K_i}=\ol{K_i}$ otherwise.
Let  $A=\{A_0,\ldots, A_{m-1}\}$ be the set of atoms 
of $L$.
By convention, $I$ is the set of initial atoms and $A_{m-1}$ is the final atom.

As above, we use a one-to-one correspondence 
$A_i \leftrightarrow  {\mathbf A}_i$ between atoms $A_i$ of a language 
$L$ and the states 
${\mathbf A}_i$ of the NFA $\cA$ defined below.
We refer to the ${\mathbf A}_i$ as \emph{atom symbols.}

\begin{definition}
\label{def:atomaton}
The \emph{\'atomaton}\footnote{In~\cite{BrTa11}, the intersection 
$A_\emp=\ol{K_0}\cap \cdots\cap \ol{K_{n-1}}$ was not considered 
an atom. It was shown that the right language of state 
${\mathbf A_i}$ is the atom $A_i$, the left language of ${\mathbf A_i}$ 
is non-empty, the language of the \'atomaton $\cA$ is $L$, and $\cA$ is trim.
If the intersection $A_\emp$ of all the complemented quotients is non-empty, 
then $A_\emp$ is an atom and $\cA$ is no longer trim
because state ${\mathbf  A}_\emp$ is not reachable from any initial state.
} of $L$
 is the NFA $\cA=({\mathbf A},\Sig,\eta, {\mathbf I},\{{\mathbf A}_{m-1}\}),$
 where ${\mathbf A}=\{{\mathbf A}_i\mid A_i\in A\}$,
 ${\mathbf I}=\{{\mathbf A}_i\mid A_i\in I\}$, 
 ${\mathbf A}_{m-1}$ corresponds to $A_{m-1}$,
 and ${\mathbf A}_j \in \eta({\mathbf A}_i, a)$ if and only if 
$aA_j \subseteq A_i$, for all ${\mathbf A_i},{\mathbf A_j}\in {\mathbf A}$ and $a\in\Sig$.
\end{definition}

\begin{example}
\label{ex:a2}
Let $L_2\subseteq \{a,c\}^*$  be defined by the quotient equations  below (left) and recognized by the DFA $\cD_2$ of Fig.~\ref{fig:automata}~(a).
The equations for the atoms of $L_2$ are below (right), and the \'atomaton $\cA_2$
is in Fig.~\ref{fig:automata}~(b); here each atom is denoted by $A_P$, where $P$ is the set of uncomplemented quotients. 
Thus $K_0\cap\ol{K_1}$ becomes $A_{\{0\}}$, etc., and we represent the sets in the subscripts without brackets and commas.
The reverse $\cD_2^\rev$ of $\cD_2$ is in Fig.~\ref{fig:automata}~(c). The determinized reverse $\cD_2^{\rev\deter}$ is 
in Fig.~\ref{fig:automata}~(d);
this is the minimal DFA for $L_2^R$, the reverse of $L_2$.
The reverse  $\cA_2^\rev$ of the \'atomaton is in Fig.~\ref{fig:automata}~(e).
Note that $\cD_2^{\rev\deter}$ and $\cA_2^\rev$ are isomorphic.
\begin{alignat*}{2}
K_0&= aK_1 \cup cK_0, & 
	\qquad  K_0\cap K_1  &= a(K_0\cap K_1) \cup c[(K_0\cap K_1)\cup (K_0\cap\ol{K_1})] ,\\  
K_1&= aK_0 \cup cK_0  \cup \eps, &
	K_0\cap \ol{K_1}&=a(\ol{K_0}\cap K_1),\\
&&
\ol{K_0}\cap K_1&=  a(K_0\cap\ol{K_1})\cup \eps,\\
& &
	\ol{K_0}\cap \ol{K_1}&=a(\ol{K_0}\cap \ol{K_1}) \cup c[(\ol{K_0}\cap \ol{K_1}) \cup 
	(\ol{K_0} \cap K_1)].
\end{alignat*}
\end{example} 

%\vspace{-.7cm}

\begin{figure}[h]
\begin{center}
\input aut.eepic
\end{center}
\caption{(a) DFA $\cD_2$; (b) \'Atomaton $\cA_2$; (c) NFA $\cD_2^\rev$; (d) 
DFA $\cD_2^{\rev\deter}$; (e)  DFA $\cA_2^\rev$.} 
\label{fig:automata}
\end{figure}

The next theorem from~\cite{Brz63}, also discussed in~\cite{BrTa11}, 
will be used several times.
\begin{theorem}[Determinization]
\label{thm:Brz}
If  an NFA $\cN$ has no empty states and $\cN^\rev$ is deterministic, 
then $\cN^\deter$ is minimal.
\end{theorem}

It was shown in~\cite{BrTa11} that the \'atomaton  $\cA$ of $L$
with reachable atoms only is isomorphic to the trimmed version of $\cD^{\rev\deter\rev}$, 
where $\cD$ is the quotient DFA of $L$.
With our new definition,  $\cA$ 
is isomorphic to $\cD^{\rev\deter\rev}$.
We now study this isomorphism in detail, along with the isomorphism between 
$\cA^\rev$ and $\cD^{\rev\deter}$.
We deal with the following automata:
\be
\item
Quotient DFA $\cD=({\mathbf K},\Sig,\delta, {\mathbf K_{in}}, {\mathbf F})$ of 
$L$ whose states are \emph{quotient symbols}.
\item
The reverse $\cD^{\rev}=({\mathbf K},\Sig,\delta^\rev,{\mathbf F},\{{\mathbf K_{in}}\})$ 
of  $\cD$. 
The states in ${\mathbf K}$ are still \emph{quotient symbols,} but their right 
languages  are no longer quotients of $L$.
\item
The determinized reverse $\cD^{\rev\deter}=(S,\Sig,\alpha, {\mathbf F}, {G})$,
where ${S}\subseteq 2^{\mathbf K}$ and
${G}=\{S_i\in S \mid {\mathbf K}_{in} \in S_i\}$.
The states in $S$ are \emph{sets of quotient symbols}, i.e., subsets of ${\mathbf K}$.
Since $(\cD^{\rev})^\rev=\cD$ is deterministic and all of its states are reachable, $\cD^\rev$ has no empty states.
By Theorem~\ref{thm:Brz}, DFA $\cD^{\rev\deter}$ is minimal and accepts $L^R$; 
hence it is isomorphic to the quotient DFA of $L^R$.
\item
The reverse $\cD^{\rev\deter\rev} = (S,\Sig,\alpha^\rev, G, \{\mathbf F\})$ of 
$\cD^{\rev\deter}$;
here the states are still \emph{sets of quotient symbols}.
\item
The \'atomaton
$\cA=({\mathbf A},\Sig,\eta, {\mathbf I},\{{\mathbf A}_{m-1}\})$,
whose states  are \emph{atom symbols.}

\item
The reverse 
$\cA^\rev=({\mathbf A},\Sig,\eta^\rev, {\mathbf A}_{m-1}, {\mathbf I})$ of $\cA$,
whose states are still \emph{atom symbols,} though their right languages are no longer atoms.
\ee
The results from~\cite{BrTa11} and  our new definition of 
atoms imply that $\cA^\rev$ is a minimal DFA that accepts $L^R$. 
It follows that  $\cA^\rev$ is  isomorphic to $\cD^{\rev\deter}$.
Our next result makes this isomorphism precise.

\begin{proposition}[Isomorphism]
\label{prop:isomorphism}
Let $\varphi: {\mathbf A} \to {S}$ be the mapping assigning to state 
${\mathbf A}_j$, given by
$A_j=K_{i_0}\cap\cdots\cap K_{i_{n-r-1}}\cap\ol{K_{i_{n-r}}}
\cap\cdots\cap \ol{K_{i_{n-1}}}$ of $\cA^\rev$, the set 
$\{K_{i_0},\ldots, K_{i_{n-r-1}}\}$.
Then $\varphi$ is a DFA isomorphism between $\cA^\rev$ and 
$\cD^{\rev\deter}$. 
\end{proposition}
\begin{proof}
The initial state ${\mathbf A}_{m-1}$ of $\cA^\rev$ is mapped to 
the set of all quotients containing $\eps$, which is precisely 
the initial state ${\mathbf F}$ of $\cD^{\rev\deter}$.
Since the quotient $L$ appears uncomplemented in every initial 
atom $A_i\in I$,  
the image $\varphi({\mathbf A}_i)$ contains $L$.
Thus the set of final states of $\cA^\rev$ is mapped to the set 
of final states of $\cD^{\rev\deter}$.

It remains to be shown,  for all ${\mathbf A}_i, {\mathbf A}_j\in {\mathbf A}$ 
and $a\in\Sig$, that $\eta^\rev({\mathbf A}_j,a)={\mathbf A}_i$ if and 
only if $\alpha(\varphi({\mathbf A}_j),a)=\varphi({\mathbf A}_i)$.

Consider atom $A_i$ with $P_i$  as the set of quotients that 
appear uncomplemented in $A_i$.
Also define the corresponding set $P_j$ for $A_j$.
If there is a missing quotient $K_h$ in the intersection
$a^{-1}A_i$, we use $a^{-1}A_i\cap (K_h\cup \ol{K_h})$.
We do this for all missing quotients until we obtain a union of atoms.
Hence ${\mathbf A}_j \in \eta({\mathbf A}_i,a)$  
 can hold in $\cA$ if and only if 
$P_j\supseteq \delta(P_i,a)$ and 
$P_j\cap \delta(Q\setminus P_i,a)=\emp$.
It follows that in $\cA^\rev$ we have $\eta^\rev({\mathbf A}_j,a)={\mathbf A}_i$
if and only if 
$P_j\supseteq \delta(P_i,a)$ and
$P_j\cap \delta(Q\setminus P_i,a)=\emp$.

Now consider $\cD^{\rev\deter}$.
Let $P_i$ be any subset of $Q$; then the successor set of $P_i$ in $\cD$ is 
$\delta(P_i,a)$. Let $\delta(P_i,a)=P_k$.
So in $\cD^{\rev}$, we have $P_i\in \delta^\rev(P_k,a)$.
But suppose that state $q$ is not in $\delta(Q,a)$; then $\delta^\rev(q,a)=\emp$.
Consequently, we also have $P_i\in \delta^\rev(P_k\cup\{q\},a)$.
It follows that for any $P_j$ containing $\delta(P_i,a)$ and satisfying 
$P_j\cap \delta(Q\setminus P_i,a)=\emp$, 
we also have $\alpha(P_j,a)=P_i$. 

We have now shown that $\eta^\rev({\mathbf A}_j,a)={\mathbf A}_i$ if and only if
$\alpha(P_j,a)=P_i$, for all subsets $P_i, P_j\in S$, that is, if and only if 
$\alpha(\varphi({\mathbf A}_j),a)=\varphi({\mathbf A}_i)$.
\qed
\end{proof}

\begin{corollary}
\label{cor:isomorphism}
The mapping $\varphi$ is an NFA isomorphism between 
$\cA$ and $\cD^{\rev\deter\rev}$.
\end{corollary}

In the remainder of the paper it is more convenient to use 
the $\cD^{\rev\deter\rev}$ representation of \'atomata, rather than 
that of Definition~\ref{def:atomaton}.

\section{The Witness Languages and Automata}
\label{sec:witness}

We now introduce a class $\{L_n\mid n\ge 2\}$ of regular languages defined by the quotient DFA's 
$\cD_n$ given below; we shall prove that the atoms of each language 
$L_n=L(\cD_n)$ in this class  meet the worst-case quotient complexity bounds.

\begin{definition}[Witness]
\label{def:witness}
For $n\ge 2$, let $\cD_n=(Q,\Sig,\delta, q_0, F)$, where $Q=\{0,\ldots,n-1\}$, 
$\Sig=\{a,b,c\}$, $q_0=0$, $F=\{n-1\}$, 
$\delta(i,a)=i+1 \mbox{ \rm mod } n$, 
$\delta(0,b)=1$, $\delta(1,b)=0$, $\delta(i,b)=i$ for $i>1$,
$\delta(i,c)=i$ for $0\le i\le n-2$, and
$\delta(n-1,c)=0$.
Let $L_n$ be the language accepted by $\cD_n$.
\end{definition}

For $n\ge 3$, the DFA of Definition~\ref{def:witness} is illustrated in Fig.~\ref{fig:witness}, and $\cD_2$ is the DFA of Example~\ref{ex:a2} ($a$ and $b$ coincide).
The DFA $\cD_n$  is minimal, since for $0\le i\le n-1$, state $i$ accepts $a^{n-1-i}$, and no other state accepts this word.
%%%%NEW
\medskip

A {\em transformation} of a set $Q$ is a mapping of $Q$ into itself. 
If $t$ is a transformation of $Q$ and  $i \in Q$, then $it$ is the {\it image} of $i$ under $t$.  
The set of all transformations of a  finite set $Q$ is a semigroup under composition, in fact, a monoid $\cT_Q$ of $n^n$ elements. 
%A~\emph{constant} transformation,  denoted by $Q \choose j$, has $it=j$ for all $i$.
A \emph{permutation} of $Q$ is a mapping of $Q$ \emph{onto} itself. 
A \emph{transposition} $(i,j)$ interchanges $i$ and $j$ and does not affect any other elements.
A~\emph{singular} transformation, denoted by $i\choose j$, has $it=j$ and $ht=h$ for all $h\neq i$.

In 1935 Piccard~\cite{Pic35} proved that three transformations of $Q$ are sufficient to generate $\cT_Q$. 
D\'enes~\cite{Den68}  studied more general generators;  we use his formulation:

\begin{theorem}[Transformations]
\label{thm:piccard1}
The  transformation monoid $\cT_Q$ can be generated by any cyclic
permutation of $n$ elements together with any transposition and any singular transformation. 
\end{theorem}

In any DFA $\cD=(Q,\Sig,\delta, q_0, F)$, each word $w$ in $\Sig^+$ performs
a transformation on $Q$ defined by $\delta(\cdot,w)$.
The set of all these transformations is the \emph{transformation semigroup} of $\cD$.
By Theorem~\ref{thm:piccard1}, the transformation semigroup of our witness $\cD_n$ has $n^n$ elements, since  $a$ is a cyclic permutation,  $b$ is a transposition and  $c$ is a singular transformation.

The following 
result of Salomaa, Wood and Yu~\cite{SWY04} concerning reversal is restated in our terminology.
\begin{theorem}[Transformations and Reversal]
\label{thm:SWY}
Let $\cD$ be a minimal DFA with $n$ states accepting a language $L$. If the transformation semigroup of $\cD$
has $n^n$ elements, then the quotient complexity of $L^R$ is $2^n$.
\end{theorem}

\begin{corollary}[Reversal]
\label{thm:reverse}
For $n\ge 2$, the quotient complexity of $L_n^R$ is $2^n$.
\end{corollary}

\begin{corollary}[Number of Atoms of $L_n$]
\label{cor:atoms}
The language $L_n$ has $2^n$ atoms.
\end{corollary}
\begin{proof}
By Corollary~\ref{cor:isomorphism}, the \'atomaton of $L_n$ is isomorphic to 
the reversed quotient DFA of $L_n^R$.
By Corollary~\ref{thm:reverse}, the quotient DFA of $L_n^R$ has $2^n$ states, 
and so the empty set of states of $L_n$ is reachable in $L_n^R$. 
Hence $L_n^R$ has the empty quotient, implying that the intersection of
all the complemented quotients  is non-empty,  
and so $L_n$ has $2^n$ atoms.
\qed\end{proof}

\begin{proposition}[Transitions of the \'Atomaton]
\label{prop:atomaton}
Let $\cD_n=(Q,\Sig,\delta, q_0, F)$ be the DFA of Definition~\ref{def:witness}.
The \'atomaton of $L_n=L(\cD_n)$ is the NFA 
$\cA_n=(2^Q,\Sig,\eta,I,\{n-1\})$, where
\be
\item If   $S=\{\emp\}$, then $\eta(S,a)=\{\emp\}$. Otherwise,\\ 
$\eta(\{s_1,\ldots, s_k\},a)=\{s_1+1, \ldots, s_k+1\}$, where the addition is modulo $n$.
\item If   $\{0,1\}\cap S=\emp$, then
	\be
	\item
	$\eta(S,b)=S$,
	\item
	$\eta(\{0\}\cup S,b)= \{1\}\cup S$,
	\item
	$\eta(\{1\}\cup S,b)= \{0\}\cup S$,
	\item
	$\eta(\{0,1\}\cup S,b)= \{0,1\}\cup S$.
	\ee
\item If   $\{0,n-1\}\cap S=\emp$, then
	\be
	\item
	$\eta(S,c)=\{S,\{n-1\}\cup S\}$,
	\item
	$\eta(\{0,n-1\}\cup S,c)= \{\{0,n-1\}\cup S, \{0\}\cup S\}$,
	\item
	$\eta(\{0\}\cup S,c)= \emp$,
	\item
	$\eta(\{n-1\}\cup S,c)= \emp$.
	\ee
\ee
\end{proposition}
\begin{proof}
The reverse  of DFA $\cD_n$ is the NFA
$\cD_n^\rev=(Q,\Sig,\delta^\rev,\{n-1\},\{0\})$, where $\delta^\rev$ is defined by
$\delta^\rev(i,a)=i-1 \mbox{ \rm mod } n$, 
$\delta^\rev(i,b)=\delta(i,b)$, 
 $\delta^\rev(0,c)=\{0,n-1\}$, 
$\delta^\rev(n-1,c)=\emp$, and $\delta^\rev(i,c)=i$, for $0<i<n-1$.
After applying 
determinization 
and reversal to $\cD_n^\rev$, the claims follow by Corollary 1.
\qed
\end{proof}

%\vspace{-.5cm}

\section{Tightness of the Upper Bounds}
\label{sec:complexity}

We now show that the upper bounds derived in Section~\ref{sec:bounds} 
are tight by proving that the atoms of the languages $L_n$ 
of Definition~\ref{def:witness} meet those bounds.

Since the states of the \'atomaton 
$\cA_n=({\mathbf A},\Sig,\eta, {\mathbf I},\{{\mathbf A}_{m-1}\})$ 
are atom symbols ${\mathbf A}_i$, and the right language of each ${\mathbf A}_i$ 
is the atom $A_i$, the languages $A_i$ are properly represented by the \'atomaton. 
Since, however, the \'atomaton is an NFA, to find the quotient complexity 
of $A_i$, we need  the equivalent minimal DFA.

Let $\cD_n$ be the $n$-state quotient DFA of Definition~\ref{def:witness} 
for $n\ge 2$, and recall that $L(\cD_n)=L_n$.
In the sequel, using Corollary~\ref{cor:isomorphism}, we represent 
the \'atomaton $\cA_n$ of $L_n$ by the isomorphic NFA
$\cD_n^{\rev\deter\rev} = (S,\Sig,\alpha^\rev, G, \{\mathbf F\})$,
and identify the atoms by their sets of uncomplemented quotients.
To simplify the notation, we represent atoms by the subscripts of the quotients, 
that is, by subsets of $Q=\{0,\ldots,n-1\}$, as in Definition~\ref{def:witness}.

In this framework, to find the quotient complexity of an atom $A_P$, 
with $P\subseteq Q$, we start with the NFA
$\cA_P=(S,\Sig,\alpha^\rev, \{P\}, \{\mathbf F\})$, which has the same states, 
transitions, and final state as the \'atomaton, but has only one initial state, 
$P$, corresponding to the atom symbol ${\mathbf A}_P$. 
Because $\cA_P^\rev$ is deterministic and $\cA_P$ has no empty states, 
$\cA_P^\deter$ is minimal by Theorem~\ref{thm:Brz}.
Therefore, $\cA_P^\deter$ is the quotient DFA of the atom $A_P$. 
The states of $\cA_P^\deter$ are certain \emph{sets of sets} of quotient symbols;
to reduce confusion we refer to them as  
\emph{collections of sets}.
The particular collections appearing in $\cA_P^\deter$ will be called ``super-algebras''.

Let $U$ be a subset of $Q$ with $|U|=u$, and let $V$ be a subset of $U$ with 
$|V|=v$. Define $\langle V \rangle_U$ to be the collection of all $2^{u-v}$ 
subsets of $U$ containing $V$. 
There are $C^n_uC^u_v$ collections of the form $\langle V \rangle_U$,
because there are $C^n_u$ ways of choosing $U$, and for each such choice 
there are $C^u_v$ ways of choosing $V$.
The collection $\langle V \rangle_U$ is called the \emph{super-algebra of $U$ 
generated by $V$}.
The \emph{type} of a super-algebra $\langle V \rangle_U$ is the ordered pair
$(|V|,|U|)=(v,u)$.

The following theorem is a well-known result of Piccard~\cite{Pic35} about 
the group---known as the \emph{symmetric group}---of all permutations of a finite set:

\begin{theorem}[Permutations]
\label{thm:piccard2}
The symmetric group of size $n!$ of all permutations of a set $Q=\{0,\ldots,n-1\}$ is generated by any cyclic permutation of 
$Q$ together with any transposition. 
%In particular, it is generated 
%by the cycle $a=(0,1,\ldots, n-1)$ and  the transposition $b=(0,1)$.
\end{theorem}

\begin{lemma}[Strong-Connectedness of Super-Algebras]
\label{lem:connectivity}
Super-algebras of the same type are strongly connected by words in $\{a,b\}^*$.
\end{lemma}
\begin{proof}
Let $\langle V_1 \rangle_{U_1}$ and $\langle V_2 \rangle_{U_2}$ be any two 
super-algebras of the same type. 
Arrange the elements of $V_1$ in increasing order, and do the same for 
the elements of the sets $V_2$, $U_1\setminus V_1$, $U_2\setminus V_2$, 
$Q\setminus U_1$, and $Q\setminus U_2$. 
Let $\pi:Q\to Q$ be the mapping that assigns the $i$th element of 
$V_2$ to the $i$th element of $V_1$, the $i$th element of
$U_2\setminus V_2$ to the $i$th element of $U_1\setminus V_1$, 
and the $i$th element of $Q\setminus U_2$ to the $i$th element of 
$Q\setminus U_1$. 
For any $R_1$ such that $V_1\subseteq R_1\subseteq U_1$, there is 
a corresponding subset $R_2=\pi(R_1)$, where $V_2\subseteq R_2\subseteq U_2$.
Thus $\pi$ establishes a one-to-one correspondence between the elements
of the super-algebras $\langle V_1 \rangle_{U_1}$ and 
$\langle V_2 \rangle_{U_2}$.
Also, $\pi$ is a permutation of $Q$, and so can be performed by a word 
$w\in \{a,b\}^*$ in $\cD_n$, in view of Theorem~\ref{thm:piccard2}.
Thus every set $R_2$ defined as above is reachable from $R_1$ by $w$. 
So $\langle V_2 \rangle_{U_2}$ is reachable from $\langle V_1 \rangle_{U_1}$.
\qed
\end{proof}

\begin{lemma}[Reachability]
\label{lem:reachability}
Let $\langle V \rangle_U$ be  any super-algebra  of type $(v,u)$.
If $v\ge 2$, then from $\langle V \rangle_U$ we can reach a super-algebra
of type $(v-1,u)$. 
If $u\le n-2$, then from $\langle V \rangle_U$ we can reach 
a super-algebra of type $(v,u+1)$.
\end{lemma}
\begin{proof}
If $v\ge 2$, then by Lemma~\ref{lem:connectivity},  
from $\langle V \rangle_U$ we can reach a super-algebra 
$\langle V' \rangle_{U'}$ of type $(v,u)$ such that $\{0,n-1\}\subseteq V'$. 
By input $c$ we reach $\langle V'\setminus\{n-1\} \rangle_{U'}$ of 
type $(v-1,u)$.
For the second claim, if $u\le n-2$, then by Lemma~\ref{lem:connectivity}, 
from $\langle V \rangle_U$ we can reach a super-algebra 
$\langle V' \rangle_{U'}$ of type $(v,u)$ such that 
$\{0,n-1\}\cap V'=\emp$.
By input $c$ we reach $\langle V' \rangle_{U'\cup \{n-1\} }$ 
of type $(v,u+1)$.
\qed
\end{proof} 

The next proposition holds for $n\ge 1$ if we let $L_1=\Sig^*$.

\begin{proposition}[Atoms with 0 or $n$ Complemented Quotients]
\label{prop:no_comp2}\\
For $n\ge 1$, the quotient complexity of the atoms $A_Q$ and $A_\emp$ of $L_n$ is $2^n-1$.
\end{proposition}

\begin{proof}
Let $\cA_Q$ ($\cA_\emp$) be the modified \'atomaton with only one initial state, $Q$ ($\emp$).
By the considerations above, $\cA_Q^\deter$ ($\cA_\emp^\deter$) is the quotient DFA of $A_Q$ ($A_\emp$); 
hence it suffices to prove the reachability of $2^n-1$ collections.

For $A_Q$, the initial state of $\cA^\deter_Q$ is the collection $\{Q\}$, which is 
the super-algebra $\langle Q \rangle_Q$ of $Q$ generated by $Q$.
Now suppose that we have reached a super-algebra of type $(v,n)$. 
 By Lemma~\ref{lem:connectivity}, we can reach every other super-algebra of type 
 $(v,n)$.
If $v\ge 2$, then by Lemma~\ref{lem:reachability} we can reach a super-algebra 
of type $(v-1,n)$. 
Thus we can reach all super-algebras $\langle V \rangle_Q$ of $Q$, one for each non-empty subset $V$ of~$Q$. 
Since there are at most $2^n-1$ collections and that many can be reached,
no other collection can be reached.

For $A_\emp$, the initial state of $\cA_\emp^\deter$ is the empty collection, which is the super-algebra $\langle \emp \rangle_\emp$ of $\emp$ generated by $\emp$.
Now suppose we have reached a super-algebra of type $(0,u)$. 
 By Lemma~\ref{lem:connectivity}, we can reach every other super-algebra of type 
 $(0,u)$.
If $u\le n-2$, then by Lemma~\ref{lem:reachability} we can reach a super-algebra 
of type $(0,u+1)$. 
Thus we can reach all super-algebras $\langle \emp\rangle_U$, one for each non-empty subset $U$ of~$Q$. 
Since there are at most $2^n-1$ collections and that many can be reached,
no other collection can be reached.

Hence the proposition holds.
\qed
\end{proof}

\begin{proposition}[Tightness]
\label{thm:cnot0}
For $n\ge 2$ and $1\le r\le n-1$, the quotient complexity of any atom of $L_n$ with $r$ complemented quotients is  $f(n,r)$.
\end{proposition}
\begin{proof}
Let $A_P$ be an atom of $L_n$ with $n-r$ uncomplemented quotients, where 
$1\le r\le n-1$, that is, let $P$ be the set of subscripts of the uncomplemented quotients. 
Let $\cA_P$ be the modified \'atomaton with the initial state $P$. 
As discussed above, $\cA_P^\deter$ is minimal; hence it suffices to prove 
the reachability of $f(n,r)$ collections.

We start with the super-algebra $\langle P \rangle_P$ with type $(n-r,n-r)$. 
By Lemmas~\ref{lem:connectivity} and \ref{lem:reachability}, we can now reach 
all super-algebras of types
\begin{eqnarray*}    
&&(n-r,n-r), (n-r-1,n-r),\ldots, (1,n-r),\\
&&(n-r,n-r+1), (n-r-1,n-r+1),\ldots, (1,n-r+1),\\
&&\hspace{4cm}\cdots\\
&&(n-r,n-1), (n-r-1,n-1),\ldots, (1,n-1).
\end{eqnarray*}

\noindent
Since the number of super-algebras of type $(v,u)$ is
$C^n_uC^u_v$,  
we can reach
$$g(n,r)=\sum_{u=n-r}^{n-1} \sum_{v=1}^{n-r} C_{u}^{n} \cdot C_{v}^{u}$$
algebras.
Changing the first summation index to $k=n-u$, we get
$$g(n,r)=\sum_{k=1}^{r} \sum_{v=1}^{n-r} C_{n-k}^{n} \cdot C_{v}^{n-k}.$$
Note that $C^n_{n-k} C^{n-k}_v=C^n_{k+v}C^{k+v}_{k}$, because 
$C^n_{n-k} C^{n-k}_v=\frac{n!}{(n-k)! k!}\cdot\frac{(n-k)!}{v! (n-k-v)!}=
\frac{n!}{k! v! (n-k-v)!}$,
and $C^n_{k+v}C^{k+v}_{k}=\frac{n!}{(k+v)! (n-k-v)!}\cdot 
\frac{(k+v)!}{k! v!}= \frac{n!}{(n-k-v)! k! v!}$. 
Now, we can write
$g(n,r)=\sum_{k=1}^{r} \sum_{v=1}^{n-r} C_{k+v}^{n} \cdot C_{k}^{k+v},$
and changing the second summation index to $h=k+v$, we have
$$g(n,r)=\sum_{k=1}^{r} \sum_{h=k+1}^{k+n-r} C_{h}^{n} \cdot C_{k}^{h}.$$
We notice that $g(n,r)=f(n,r)-1$.
From the super-algebra $\langle V \rangle_V$, where $V=\{0,1,\ldots,n-r-1\}$, 
we reach the empty quotient by input $c$, since $V$ contains 0, but not $n-1$.

Since we can reach $f(n,r)$ super-algebras, no other collection 
can be reached, 
and the proposition holds.
\qed
\end{proof}

The entire process of finding the complexity of atoms is illustrated in the example below for $n=3$.
\begin{example}
Let $L_3$ be the language accepted by the quotient DFA 
$\cD_3$ of Definition~\ref{def:witness} and Table~\ref{tab:d3}, 
where the initial state is identified by an incoming arrow and  
the final state, by an outgoing arrow. The first column consists of 
states $q$, and the remaining columns give the values of $\delta(q,x)$ 
for each $x\in \Sig$.
Let the   quotients of $L_3$ be
$K_0=L_3=\eps^{-1}L_3$, $K_1=a^{-1}L_3$, and $K_2=(aa)^{-1}L_3$.
The states of $\cD_3$ are subscripts of quotient symbols.

Reversing $\cD_3$, we obtain the NFA $\cD_3^{\rev}$ of 
Table~\ref{tab:d3rev}.
The states of $\cD_3^\rev$ are the same as those of $\cD_3$, 
but the transitions are to sets of states, and 
$02$ stands for $\{0,2\}$, $0$ stands for $\{0\}$, etc.

%%%%%%%%

\begin{table}[hbt]

\begin{minipage}[b]{0.45\linewidth}

%MY TABLE 1
\caption{Quotient DFA $\cD_3$ of $L_3$.}
\label{tab:d3}
\begin{center}
$
\begin{array}{|c| c||c| c| c|c|}    
\hline
& \ \  \ \ 
&\ \ a \ \ &\ \ b \ \ & \ \ c \ \ & \\
\hline  
\rightarrow & 0
& 1 & 1  & \ 0 \  & \\
\hline  
& 1
& 2 & 0  & \ 1 \ & \\
\hline  
& 2
&  0 &  2 & \ 0  \ &\rightarrow \\
\hline  
\end{array}
$
\end{center}
%END MY TABLE 1

\end{minipage}
\hspace{0.2cm}
\begin{minipage}[b]{0.45\linewidth}

% MY TABLE 2
\caption{NFA $\cD_3^\rev$ for $L_3^R$.}
\label{tab:d3rev}
\begin{center}
$
\begin{array}{|c| c||c| c| c|c|}    
\hline
& \ \  \ \ 
&\ \ a \ \ &\ \ b \ \ &  \ c \ & \\
\hline  
 & 0
& 2 & 1  & \ 02 \  & \rightarrow \\
\hline  
& 1
& 0 & 0  & \ 1 \ & \\
\hline  
\rightarrow & 2
&  1 &  2 & \ \emp  \ & \\
\hline  
\end{array}
$
\end{center}
%END MY TABLE 2

\end{minipage}
\end{table}

Next, we perform the subset construction on $\cD_3^\rev$ to determinize it and get the DFA $\cD_3^{\rev\deter}$, the quotient DFA for  $L_3^R$.
Since $\cD_3^\rev$ is trim, and $(\cD_3^\rev)^\rev=\cD_3$ is deterministic, the DFA $\cD_3^{\rev\deter}$ shown in Table~\ref{tab:d3revdet} is minimal by Theorem~\ref{thm:Brz}.

The states of $\cD_3^{\rev\deter}$ are sets of (subscripts of) quotient symbols.
 Now we reverse $\cD_3^{\rev\deter}$  to get $\cD_3^{\rev\deter\rev}$ of Table~\ref{tab:atomaton}, which is
 isomorphic to 
 the \'atomaton 
$\cA_3$.
The states of $\cD_3^{\rev\deter\rev}$  are still sets of (subscripts of) quotient symbols.
Note that the empty set $\emp$ of quotient symbols is a state of $\cD_3^{\rev\deter}$, and hence also of $\cA_3$.
It is not to be confused with the empty set of transitions associated with states 0, 2, 01, and 12 under input $c$ indicated by $-$.
 \begin{table}[hbt]

\begin{minipage}[b]{0.45\linewidth}

%MY TABLE 1
\caption{ DFA $\cD_3^{\rev\deter}$ for $L_3^R$.}
\label{tab:d3revdet}
\begin{center}
$
\begin{array}{|c| c||c| c| c|c|}    
\hline
& \ \  \ \ 
&\ \ a \ \ &\ \ b \ \ &  \ c \  & \\
\hline  
& \emp
& \emp & \emp  & \ \emp \ &  \\
\hline  
& 0
&  2 &  1 & \ 02  \ &\rightarrow \\
\hline  
& 1
& 0 & 0  & \ 1 \ &  \\
\hline  
\rightarrow & 2
& 1 & 2  & \ \emp \  & \\
\hline  
& 01
&  02 &  01 & \ 012  \ &\rightarrow \\
\hline  
& 02
&  12 &  12 & \ 02  \ &\rightarrow \\
\hline  
& 12
&  01 &  02 & \ 1  \ & \\
\hline  
&\ 012\
&\  012\ & \ 012\ & \ 012  \ &\rightarrow \\
\hline  
\end{array}
$
\end{center}
%END MY TABLE 1

\end{minipage}
\hspace{0.3cm}
\begin{minipage}[b]{0.45\linewidth}

% MY TABLE 2
\caption{\'Atomaton $\cA_3=\cD_3^{\rev\deter\rev}$.}
\label{tab:atomaton}
\begin{center}
$
\begin{array}{|c| c||c| c| c|c|}    
\hline
& \ \  \ \ 
&\ \ a \ \ &\ \ b \ \ &  \ c \  & \\
\hline  
 & \emp
&  \emp &  \emp & \ \emp,2 \ &  \\
\hline  
\rightarrow & 0
&  1 &  1 & \ -  \ &  \\
\hline  
& 1
&  2 &  0 & \ 1,12  \ & \\
\hline  
& 2
&  0 &  2 & \ -  \ &\rightarrow \\
\hline  
\rightarrow & 01
& 12 & 01  & \ - \ &  \\
\hline  
\rightarrow  & 02
& 01 & 12  & \ 0,02 \ &  \\
\hline  
& 12
&  02 &  02 & \ -  \ & \\
\hline  
\rightarrow & \ 012 \
&\ 012\ &\ 012\  & \ 01,012 \  & \\
\hline  
\end{array}
$
\end{center}

%END MY TABLE 2

\end{minipage}

\end{table}

Atom $A_{012}=K_0\cap K_1\cap K_2$ is the language accepted by $\cA_3$ started in state $012$.
The states of $\cD_{012}$, the minimal DFA of $A_{012}$, are collections of sets of quotients.
As seen from Table~\ref{tab:a012}, the quotient complexity of $A_{012}$ is seven.
\begin{table}[ht]
\caption{DFA $\cD_{012}$ of $A_{012}$.}
\label{tab:a012}
\begin{center}
$
\begin{array}{|c| c||c| c| c|c|}    
\hline
& \ \  \ \ 
&\ \ a \ \ &\ \ b \ \ &  \ c \  & \\
\hline  
\rightarrow & 012
& 012 & 012  & \ 01,012 \  & \\
\hline  
&\ 01,012\
& \ 12,012 \ & \ 01,012\  & \ 01,012 \ &  \\
\hline  
& 02,012
&  01,012 &  12,012 & \ 0,01,02,012  \ &  \\
\hline  
 & 12, 012
& 02,012 & 02,012  & \ 01,012 \ &  \\
\hline  
& \ 0, 01,02,012 \
&  \ 1,01,12,012\  &  \ 1,01,12,012\ & \ 0,01,02, 012 \ & \\
\hline  
& \ 1,01,12,012\
&  \ 2,02,12,012\ &  0,01,02,012 & \ 1,01,12,012  \ & \\
\hline  
& \ 2,02,12,012\
&  0,01,02,012  &  2,02,12,012 & \ 0,01,02,012  \ &\rightarrow \\
\hline  
\end{array}
$
\end{center}
\end{table}

Atom $A_{01}=K_0\cap K_1\cap \ol{K_2}$ is  accepted by $\cA_3$ started in  state $01$.
The minimal DFA $\cD_{01}$ of $A_{01}$ is shown in Table~\ref{tab:a01}, and the quotient complexity of $A_{01}$ is ten.
Since $01$, $12$ and $02$ are strongly connected by $a$,
the same collections are reached  from these states,  and so the quotient complexity of $A_{12}$ and $A_{02}$ is also ten.

Atom $A_2=\ol{K_0}\cap \ol{K_1}\cap K_2$ is  accepted by $\cA_3$ started in state $2$.
The minimal DFA $\cD_{2}$ of $A_2$ is shown in Table~\ref{tab:a2}, and the quotient complexity of $A_2$ is ten.
Since $0$, $1$ and $2$ are strongly connected by $a$,
the same collections are reached  from these states,  and so the quotient complexity of $A_{0}$ and $A_{1}$ is also ten.

\begin{table}[h]
\begin{minipage}[b]{0.45\linewidth}
%MY TABLE 1
\caption{DFA $\cD_{01}$ of $A_{01}$.}
\label{tab:a01}
\begin{center}
$
\begin{array}{|c| c||c| c| c|c|}    
\hline
& \ \  \ \ 
&\ \ a \ \ &\ \ b \ \ &  \ c \  & \\
\hline  
& \emp
& \emp & \emp  & \ \emp \  & \\
\hline  
\rightarrow & 01
& 12 & 01 & \ \emp \  & \\
\hline
 & 02
& 01 & 12 & \ 0,02 \  & \\
\hline
 & 12
& 02 & 02 & \ \emp \  & \\
\hline  
& 0,01
& 1,12 & 1,01  & \ \emp \  & \\
\hline  
& 0,02
& 1,01 & 1,12  & \ 0,02 \  & \\
\hline  
& 1,01
& 2,12 & 0,01  & \ 1,12 \  & \\
\hline  
& 1,12
& 2,02 & 0,02  & \ 1,12 \  & \\
\hline  
&\ 2,02 \
&\ 0,01 \ &\ 2,12 \  & \ 0,02 \  & \rightarrow \\
\hline  
& 2,12 
& 0,02  & 2,02  & \ \emp \  & \rightarrow\\
\hline  
\end{array}
$
\end{center}

%END MY TABLE 1

\end{minipage}
%\hspace{0.3cm}
\begin{minipage}[b]{0.45\linewidth}

% MY TABLE 2
\caption{DFA $\cD_{2}$ of $A_2$.}
\label{tab:a2}
\begin{center}
$
\begin{array}{|c| c||c| c| c|c|}    
\hline
& \ \  \ \ 
&\ \ a \ \ &\ \ b \ \ &  \ c \  & \\
\hline  
& \emp
& \emp & \emp  & \ \emp \  & \\
\hline  
& 0
& 1 & 1  & \ \emp \  & \\
\hline  
 & 1
& 2 & 0  & \ 1,12 \  & \\
\hline  
\rightarrow & 2
& 0 & 2  & \ \emp  \  & \rightarrow\\
\hline  
 &\ 0,01 \
& 1,12  & 1,01  & \ \emp \  & \\
\hline  
 & 0,02 
& 1,01 & 1,12  & \ 0,02  \  & \\
\hline  
 & 1,01 
&\ 2,12  \ &\ 0,01 \  & \ 1,12 \  & \\
\hline  
 & 1,12
& 2,02 & 0,02  & \ 1,12 \  & \\
\hline  
 & 2,02 
& 0,01 & 2,12  & \ 0,02 \  &\rightarrow \\
\hline  
 & 2,12 
& 0,02 & 2,02   & \ \emp \  &\rightarrow \\
\hline  
\end{array}
$
\end{center}
%END MY TABLE 2
\end{minipage}
\end{table}

Finally, atom $A_\emp=\ol{K_0}\cap\ol{K_1}\cap\ol{K_2}$ is accepted by $\cA_3$ started in state $\emp$.
The minimal DFA $\cD_\emp$ is shown in Table~\ref{tab:aemp}, and the quotient complexity of $A_\emp$ is seven.
Note that $\cD_{012}$ and $\cD_\emp$ have isomorphic transition tables, if we ignore final states.
The isomorphism is $\psi:2^{2^Q}\to 2^{2^Q}$ defined as follows:
If $C\subseteq 2^Q$ is a collection of subsets of $Q$, then $\psi(C)=\{Q\setminus S\mid S\in C\}$.
%\qedb
%\end{example} 

\begin{table}[ht]
\caption{DFA $\cD_\emp$ of $A_\emp$.}
\label{tab:aemp}
\begin{center}
$
\begin{array}{|c| c||c| c| c|c|}    
\hline
& \ \  \ \ 
&\ \ a \ \ &\ \ b \ \ &  \ c \  & \\
\hline  
\rightarrow & \emp
& \emp & \emp  & \ \emp,2 \  & \\
\hline  
& \emp,0
&  \emp,1 &  \emp,1 & \ \emp,2  \ &  \\
\hline  
 & \emp,1
& \emp,2 & \emp,0  & \ \emp,1,2,12 \ &  \\
\hline
&\ \emp,2\
& \ \emp,0 \ & \ \emp,2\  & \ \emp,2 \ & \rightarrow \\
\hline  
& \ \emp,0,1,01\
&  \emp,1,2,12  &  \emp,0,1,01 & \ \emp,1,2,12  \ & \\
\hline  
& \ \emp,0,2,02\
&  \ \emp,0,1,01\ &  \emp,1,2,12 & \ \emp,0,2,02  \ & \rightarrow\\
\hline  
& \ \emp, 1,2,12 \
&  \ \emp, 0,2,02 \  &  \ \emp, 0,2,02 \ & \ \emp,1,2,12 \ & \rightarrow\\
\hline  
\end{array}
$
\end{center}
\end{table}
\end{example}

\section{Conclusions}
\label{sec:conclusions}
The atoms of a regular language $L$ are its basic building blocks. 
We have studied the quotient complexity of the atoms of $L$ as a function 
of the quotient complexity of $L$.
We have computed an upper bound for the quotient complexity of any atom 
with $r$ complemented quotients, and exhibited a class $\{L_n\}$ of 
languages whose atoms meet this bound.
\bigskip

\noin
{\bf Acknowledgments}
We are grateful to Baiyu Li for writing a program for evaluating 
the quotient complexity of atoms.
We thank Eric Rowland and Jeff Shallit 
for computing the ratio defined for Table~\ref{tab:atomcomp} 
for some large values of $n$.

%\bibliography{atomata}
%\bibliographystyle{splncs_srt2}
\end{document}